\documentclass[copyright,creativecommons]{eptcs}
\providecommand{\event}{QAPL 2013}
\usepackage[T1]{fontenc}
\usepackage{times}
\usepackage{hyperref} 
\usepackage{amsfonts}
\usepackage{wrapfig} 
\usepackage{amsmath}
\usepackage{amsthm}
\usepackage{amssymb}
\usepackage{wasysym}
\usepackage{nicefrac}
\usepackage[strings]{underscore}
\usepackage{tikz}
\usetikzlibrary{arrows,shapes}
\newtheorem{remark}{Remark}
\newtheorem{example}{Example}[section]
\newtheorem{theorem}{Theorem}[section]

\newtheorem{definition}{Definition}[section]
\newtheorem{mechanism}{Mechanism}
\newtheorem{hypo}{Condition}
\renewenvironment{proof}{\noindent {\bf Proof}\quad}{\qed}

\newcommand{\lbg}{\lambda}
\newcommand{\sa}{\mathcal{S}}
\newcommand{\R}{\mathbb{R}}

\newcommand{\M}{{\mathbb{R}^m}}
\newcommand{\dimension}{m}
\newcommand{\Mr}{\mathbb{M}_r}
\newcommand{\U}{{U^q}}
\newcommand{\dimU}{q}
\newcommand{\Ud}{[0,1]^\dimU}
\newcommand{\Ur}{U_r}
\newcommand{\D}{\mathcal{D}}
\newcommand{\Pb}{P}
\newcommand{\diam}{\diameter}

\newcommand{\nmax}{dn_{\text{max}}}
\newcommand{\extr}{\infty}

\newcommand{\ratio}{R}
\newcommand{\arrondi}{L}

\newcommand{\AND}{\mathrel{\wedge}}

\title{Preserving differential privacy under finite-precision semantics 
\thanks{This work has been partially supported by the project
        ANR-09-BLAN-0345-02  CPP and by the INRIA Action d'Envergure CAPPRIS.}}
\def\titlerunning{Preserving differential privacy}
\author{Ivan Gazeau, Dale Miller, and Catuscia Palamidessi \\
INRIA and LIX, Ecole Polytechnique}
\def\authorrunning{I. Gazeau, D. Miller, and C. Palamidessi}
\begin{document}
\maketitle

\begin{abstract}
The approximation introduced by finite-precision
representation of continuous data can induce arbitrarily large
information leaks even when the computation using exact semantics is
secure. Such leakage can thus undermine design efforts aimed
at protecting sensitive information.  We focus here on differential
privacy, an approach to privacy that emerged from the area of
statistical databases and is now widely applied also in other domains.
In this approach, privacy is protected by the addition of noise to a
true (private) value.  To date, this approach to privacy has been
proved correct only in the ideal case in which computations are made
using an idealized, infinite-precision semantics.  In this paper, we
analyze the situation at the implementation level, where the semantics
is necessarily finite-precision, i.e. the representation of real
numbers and the operations on them, are rounded according to some
level of precision. We show that in general there are violations of
the differential privacy property, and we study the conditions under
which we can still guarantee a limited (but, arguably, totally
acceptable) variant of the property, under only a minor degradation of
the privacy level. Finally, we illustrate our results on two cases of
noise-generating distributions: the standard Laplacian mechanism
commonly used in differential privacy, and a bivariate version of the
Laplacian recently introduced in the setting of privacy-aware
geolocation.
%
%

\end{abstract}

\section{Introduction}

It is well known that, due to the physical limitations of actual machines, 
in particular the finiteness of their memory, real numbers and their 
operations cannot be implemented with full precision.  
While for traditional computation getting an approximate result is
not critical when a bound on the error is known, we argue that, 
in  security applications, the approximation 
error can became a fingerprint potentially causing the disclosure of secrets. 

Obviously, the standard techniques to measure the security breach do not apply, 
because an analysis of the system in the \emph{ ideal } (aka \emph{exact}) semantics
does not reveal the information leaks caused by the implementation. 
Consider, for instance, the following simple program
\[ \mbox{if} \;\; f(h) > 0\;\; \mbox{then}\;\; \ell = 0 \;\; \mbox{else} \;\; \ell = 1\]
where $h$ is a high (i.e., confidential) variable and $\ell$ is a low (i.e., public) variable. 
Assume that $h$ can take two values, $v_1$ and $v_2$, and that both $f(v_1)$ and $f(v_2)$ 
are strictly positive. Then, in the ideal semantics, the program is perfectly secure, 
i.e. it does not leak any information. 
However, in the implementation, it could be the case that the test succeeds in the case of 
$v_1$ but not in the case of $v_2$ because, for instance, the value of 
$f(v_2)$ is below  the smallest representable positive number. 
Hence, we would have a  total disclosure of the secret value. 

The example above is elementary but it should give an idea of the 
pervasive nature of the problem, which can have an impact in any confidentiality setting, 
and should therefore receive attention by those researchers interested in (quantitative) information flow.  
In this paper, we initiate this investigation with an in-depth study of the particular case of 
\emph{differential privacy}.

Differential privacy \cite{Dwork:06:ICALP,Dwork:06:TCC} is an approach
to the protection of private information that originated in the field
of statistical databases and is now investigated in many other
domains, ranging from programming languages
\cite{Barthe:12:POPL,Gaboardi:13:POPL} to social networks
\cite{Narayanan:09:SSP} and geolocation
\cite{Machanavajjhala:08:ICDE,Ho:11:GIS,Mig12,geoPets13}.
The key idea behind differential privacy is that whenever someone queries a dataset, the 
reported  answer should not allow him to distinguish  
whether a certain individual record is in the dataset or not.  
More precisely, the presence or absence of the record should not 
change significantly the probability of obtaining a given answer. 
The standard way of achieving such a property is by using an 
\emph{oblivious mechanism}\footnote{The name ``oblivious'' comes from the fact that 
the final answer depends only on the answer to the query and not on the dataset.} 
which consists in adding some noise to the true answer. 
Now the point is that, even if such a mechanism is proved to provide the desired 
property in the ideal semantics, its  implementation may induce errors
that  alter the least significant digits of the reported answer and 
cause significant privacy breaches. Let us illustrate the  problem with an example.
\begin{example}\label{exa:holes in the Laplacian}
{\rm 
Consider the simplest representation of reals: the fixed-point numbers. 
This representation is used on low-cost processors which do not have floating-point arithmetic module.
Each value is stored in a memory cell of fixed length. In such cells, the last $d$ digits represent the fractional part.
Thus, if the value (interpreted as an integer) stored in the cell is $z$, its semantics (i.e., the true real number being represented) is $z\cdot2^{-d}$.
 
To grant differential-privacy, the standard technique consists of returning a random value with probability
$p(x)=\nicefrac{1}{2b}\cdot e^{-\nicefrac{|x-r|}{b}}$ 
where $r$ is the true result and $b$ is a scale parameter which depends on the degree of privacy to be obtained and on the sensitivity of the query.
To get a random variable with any specific distribution, in general, we need to start with an initial random variable 
provided by a primitive of the machine with a given distribution. 
To simplify the example, we assume that the machine already provides a Laplacian random variable $X$ with a scale parameter $1$.
The probability distribution of such an $X$ is $p_X(x)= \nicefrac{1}{2} e^{-|x|}$.
Hence, if we want to generate the random variable $bX$ with probability distribution $p_{bX}(x)= \nicefrac{1}{2b}\cdot e^{-\nicefrac{|x|}{b}}$,
we can just multiply by $b$ the value $x=z\cdot 2^{-d}$ returned by the primitive.

Assume that we want to add noise with a scale parameter $b=2^n$ for some fixed integer $n$ 
($b$ can be big when the sensitivity of the query and the required privacy degree are high).
In this case, the multiplication by $2^n$ returns a number  $2^n z\cdot2^{-d}$ that, in the fixed-point representation terminates with $n$ zeroes.
Hence, when we add this noise to the true result, we return a value whose representation has the same $n$ last digits as the secret.
For example, assume $b=2^2=4$ and $d=6$.
Consider that the true answers are $r_1=0$ and $r_2=1+2^{-5}$.
In the fixed-point representation, the last two digits of $r_1$ are $00$, and the last two digits of $r_2$ are $10$. 
Hence, even after we add the noise, we can always tell whether the true value was $r_1$ or $r_2$.
Note that the same example holds for every $b=2^n$ and every pair of true values $r_1$ and $r_2$ which differ by $\nicefrac{(2^n k+h)}{2^d}$
where $k$ is any integer and $h$ is any integer between $1$ and $2^n -1$. 
Figure\ \ref{ex1} illustrates the situation for $b=4$, $k=3$ and $h=2$. 
}
\end{example}

\begin{figure}
\centering
 \newcommand*{\sensi}{0.5}
 \newcommand*{\step}{0.06}
 \newcommand*{\coef}{3}

\begin{tikzpicture}[domain=-4:4]
    \draw[->] (-5.5,0) -- (5.5,0) node[right] {$x$};
    \draw[->] (0,-0.2) -- (0,3.7) node[above] {$P(x)$};
\node at (0,-0.5){$0$};
\foreach \x in {-20,...,20}{

\path[fill=black] (\x * 4 * \step ,0 )
 rectangle (\x * 4 * \step + \step ,{round(\coef *exp(-abs((\x + 0)*\step *4))/\step)*\step});

\path[fill=green] (\x * 4 * \step + \step * 2  ,0 )
 rectangle (\x * 4 * \step + \step * 3, {round(\coef *exp(-abs((\x  - 3)*4*\step))/\step)*\step});

}
\end{tikzpicture}

\caption{The probability distribution of the reported answers after the addition of Laplacian noise for the true answer $r_1=0$ (black) and
$r_2=3\cdot2^{-4}+2^{-5}$ (green).}\label{ex1}
\end{figure}
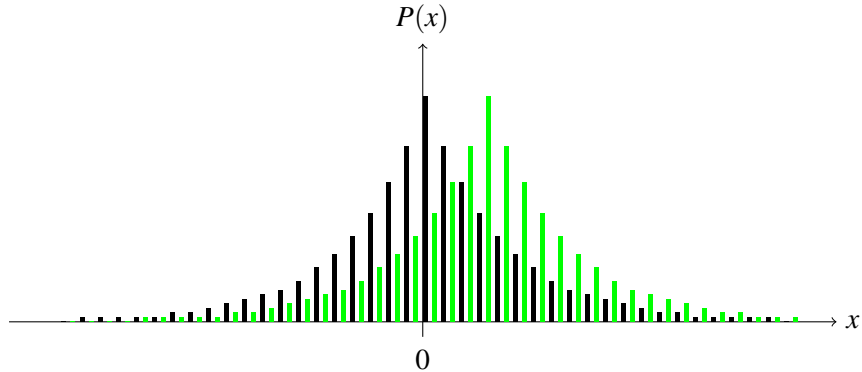

Another attack, based on the IEEE standard floating-point representation \cite{ieee08}, was presented in \cite{Mironov12}.
In contrast to \cite{Mironov12}, we have chosen an example based on the fixed point representation because  
it allows to illustrate more distinctively a problem for privacy which rises from the finite 
precision\footnote{More precisely, the problem is caused by scaling a finite set of randomly generated numbers. It is easy to prove that the 
problem raises for any implementation of numbers, although it may not raise \emph{for every point} like in the case of the fixed-point representation.}
 and which is, therefore,
pandemic. (This is not the case for the example in \cite{Mironov12}: 
fixed-point and integer-valued algorithms are immune to that
attack.)

In this paper, we propose a solution to fix the privacy breach induced
by the finite-precision implementation of a differentially-private
mechanism for any kind of implementation.  
Our main concern is to establish a bound on the degradation of  privacy induced by 
both the finite representation and  by the computational errors in the generation of the noise. 
In order to achieve this goal, we use the concept of \emph{closeness}
introduced by the authors in \cite{gazeau12hal}, which allows us to reason about the approximation errors and their accumulation. 
We make as few assumptions  as possible about the procedure for generating the noise.
In particular,  we do not assume that the noise has a linear Laplacian distribution: it can be any noise that provides differential privacy
and whose implementation satisfies a few properties (normally granted by the implementation of real numbers) which ensure its closeness. 
We illustrate our method with two examples: the classic case of the
univariate (i.e., linear) Laplacian, and  the case of the bivariate Laplacian. The latter distribution is used,
for instance, to generate noise in privacy-aware geolocation mechanisms \cite{Mig12}.   

\subsection{Related work}
As far as we know, the only other work that has considered the problem
introduced by the finite precision in the implementation of
differential privacy is \cite{Mironov12}. As already mentioned, that
paper showed an attack on the Laplacian-based
implementation of differential privacy within the IEEE standard
floating-point representation\footnote{We discovered our attack
  independently, but \cite{Mironov12} was published first.}. 
To thwart such an attack, the author of \cite{Mironov12} proposed a method
that avoids using the standard uniform random generator for floating point
(because it does not draw all representable numbers but only multiple of $2^{-53}$). 
Instead, his method generates two integers, one for the mantissa and one for the exponent 
in such a way that every representable number is drawn with its correct probability.
Then it computes the linear Laplacian using a logarithm implementation (assumed to be full-precision), 
and finally it uses a snapping mechanism consisting in truncating large values and then rounding the final result.

The novelties of our paper, w.r.t. \cite{Mironov12}, consist in the fact that 
we deal with a general kind of noise, not necessarily the linear Laplacian, 
and with any kind of  implementation of real numbers, not necessarily the IEEE floating point standard. 
Furthermore, our kind of analysis allows us to measure how safe an existing solution can be and 
what to do if the requirements needed for the safety of this solution are not met.
Finally, we consider our correct implementation of  the bivariate Laplacian also as a valuable contribution, 
given its practical usefulness for location-based applications.  

The only other work we are aware of, considering both computational error and differential  privacy, is   \cite{Chaudhuri:11:FSE}. 
However, that paper does not consider at all the problem of the loss of privacy due to implementation error:                                                                                                                                                                                                                                                                                                                                                                                                                                                                                                                                
rather, they develop a technique to establish a bound on the error, and show that this technique can also be used to 
compute the sensitivity of a query, which is a parameter of the Laplacian noise. 

\subsection{Plan of the paper}
This paper is organized as follow. In section \ref{sec:notation}, we recall some mathematical definitions and introduce some notation.
In section \ref{sec:diff-priv}, we describe the standard Laplacian-based mechanism 
that provides differential privacy in a theoretical setting.
In section \ref{sec:errors}, 
we discuss the  errors due to the implementation,
and we consider a set of assumptions which, if granted, allows us to 
establish a  bound on the irregularities  of the noise caused by the finite-precision 
implementation. 
Furthermore we propose a correction to the mechanism
based on rounding and truncating the result. 
Finally, we provide our main theorem, stating that with
our correction the implementation of the mechanism still preserves differential privacy, 
and establishing the precise  degradation of the privacy parameter. 
The next two sections propose some applications of  our result:
Section \ref{sec:one-dim} illustrates the  technique for the case of 
Laplacian noise in one dimension and section \ref{sec:two-dim} 
shows how our theorem applies to the case of the Euclidean bivariate Laplacian.
Section~\ref{sec:conclusion} concludes and discusses some future work. 

For the sake of space, the proofs are omitted. The interested reader can find them in the full version of this paper~\cite{gazeau13hal}.

\section{Preliminaries and notation}\label{sec:notation}

In this section, we recall some basic mathematical definitions and we 
introduce some notation that will be useful in the rest of the paper.
We will assume that the  the queries give answers in $\M$. 
Examples of such queries are the tuples representing, for instance, the average height, weight, and age.
Another example comes from geolocation, where the domain is $\R^2$. 

\subsection{Distances and geometrical notations}

There are several natural definitions of distance on  $\M$ \cite{rudin86book}.  For $\dimension \in \mathbb{N}$ and
$x=(x_1,\dots,x_\dimension) \in \M$, the $L_p$ norm of $x$, which we will denote by $\|x\|_p$, is defined as
$\|x\|_p=\sqrt[p]{\sum_{i=1}^\dimension |x_i|^p}$.
The corresponding distance function is $d_p(x,y)=\|x-y\|_p$.
We extend this norm and distance to
$p=\infty$ in the usual way:  $\|x\|_\infty=\max_{i\in\{1,\ldots,\dimension\}}
|x_i|$ and $d_\infty(x,y)=\|x-y\|_\infty$.
The notion of $L_\infty$ norm is extended to functions in the following way:  given $f : A \to
\M$, we define $\|f\|_\infty = \max_{x\in A} \|f(x)\|$.
When clear from the context, we will omit the parameter $p$ and write simply 
$\|x\|$ and $d(x,y)$  for
$\|x\|_p$ and $d_p(x,y)$, respectively.

Let $S \subseteq \M$.  We denote by $S^c$ the \emph{complement of} $S$, i.e.,
$S^c=\M \setminus S$.  The \emph{diameter of} $S$ is defined as  
\[\diam(S)= \max_{x,y \in S} d(x,y).\]  
For $\epsilon \in \R^+$, the $+\epsilon$-\emph{neighbor}  and the  $-\epsilon$-\emph{neighbor} of $S$ are defined as
\[S^{+\epsilon}=\{ x \;|\; \exists s \in S, d(x,s) \leq \epsilon \}   \qquad
S^{-\epsilon}=\{ x \;|\; \forall s \in \M, d(x,s) \leq \epsilon 
  \implies s \notin S \}=((S^c)^{+\epsilon})^c \] 
For $x \in \M$, the \emph{translations} of $S$ by $x$ and $-x$ are defined as
\[ S+x= \{y+x \;|\; y \in S \} \qquad S-x= \{y-x \;|\; y
\in S \}\]

\subsection{Measure theory}

We recall here some basic notions of measures theory that will be
used in this paper. 

\begin{definition}[$\sigma$-algebra and measurable space]
A $\sigma$-algebra $\mathcal{T}$ for a set $M$ is a nonempty set of
subsets of $M$ that is closed under complementation (wrt to $M$) and (potentially empty)
 enumerable union. The tuple $(M,\mathcal{T})$ is called 
a \emph{measurable space}. 
\end{definition}

 
\begin{definition}[Measure]
A {\em positive measure} $\mu$ on a measurable space $(M,\mathcal{T})$ is a
function $\mathcal{T} \to \R^+ \cup \{0\}$ such that
$\mu(\emptyset)=0$ and whenever $(S_i)$ is a enumerable family of
disjoint subset of $M$ then  
$\sum \mu(S_i) = \mu(\bigcup S_i)$.
A positive measure $\mu$ where $\mu(X)=1$ is called a \emph{probability measure}.
\end{definition} 

A tuple $(M,\mathcal{T},\Pb)$ where $(M,\mathcal{T})$ is a
measurable space and $\Pb$ a probability measure is called
\emph{probability space}.

In this paper we will make use of the Lebesgue measure $\lbg$  on $(\M,\sa)$ where $\sa$ is
the Lebesgue $\sigma$-algebra. The Lebesgue measure is the standard
way of assigning a measure to subsets of $\M$.                                                                           

\begin{definition}[Measurable function]
Let  $(M,\mathcal{T})$ and $(V,\Sigma)$ be two measurable spaces.  A
function $f : M \to V $ is measurable if $f^{-1}(v) \in \mathcal{T}$
for all $v \in \Sigma$.
\end{definition}

\begin{definition}[Absolutely continuous]
A measure $\nu$ is absolutely continuous according to a measure $\mu$, if for all $M \in \sa$,
$\mu(M)=0$ implies $\nu(M)=0$.  
\end{definition}
If a measure is absolutely continuous according to the Lebesgue measure then by the Radon-Nikodym theorem,
 we can express it as an integration of a density function $f$:
$\mu(M)=\int_M f(x) \, d\lbg$.


\subsection{Probability theory}

\begin{definition}[Random variable]
Let $(\Omega, \mathcal{F}, \Pb)$ be a probability space and $(E, \mathcal{E})$ a measurable space.
Then a random variable is a measurable function  $X\ :\ \Omega\ \to E$. 
We shall use the expression $\Pb\left[X\in B\right]$ to denote
$\Pb\left(X^{-1}(B)\right)$. 
\end{definition}

Let $f : \M \to \M$ be a measurable function and let $X\ :\ \Omega\ \to \M$ be a random variable.
In this paper, we will use the notation $f(X)$ to denote 
the random variable $Y\ :\ \Omega\ \to \M$ such that $f(X)(\omega)=f(X(\omega))$.
In particular, for $m \in \M$ we denote by $m+X$ the random variable $Y : \Omega\ \to \M$ such that $\omega \mapsto X(\omega) + m$.  

\begin{definition}[Density function]
Let $X\ :\ \Omega \to E$ be a random variable.
If there exists a function $f$ such that, for all $S \in \sa$,
$\Pb[X \in S]=\int_{S} f(u) \, du$, \; then $f$ is called the \emph{density function} of $X$.
\end{definition}

In this paper, we use the following general definition of the Laplace distribution (centered at zero).
\begin{definition}[Laplace distribution]\label{def:Laplace}
The density function $F$ of a Laplace distribution with scale parameter $b$ is
$F_b(x)=K(b) e^{-b\|x\|}$ where $K(b)$ is a normalization factor which is  determined by 
imposing $\int_{S} F_b(x)\, dx =1$.  
\end{definition}

\begin{definition}[Joint probability]
Let $(X,Y)$ be a pair of random variable on $\M$. The joint probability  on $(X,Y)$ is defined for all $I, J \in \sa$ as:
$\Pb[(X,Y) \in (I,J)] = \Pb[X \in I \AND Y \in J] $.
\end{definition}

\begin{definition}[Marginals]
Let $X$ and $Y$ :  $(\Omega, \mathcal{F}, \Pb) \to (\M,\sa)$.
The marginal probability of the random variable $(X,Y)$ for $X$ is defined as:
$\Pb[X\in B] = \int_{\M} \Pb[{(X, Y)} \in (B, \mathrm dy)]$.
\end{definition}

\section{Differential privacy in the exact semantics}\label{sec:diff-priv}
In this section, we recall the definition of differential privacy and  of the standard 
mechanisms to achieve it, and we discuss its correctness. 

\subsection{Differential privacy}

We denote by $\D$ the set of databases and we assume that the domain
of the answers of the query is $\M$ for some $n \geq 1$.  We denote by
$D_1 \sim D_2$ the fact that $D_1$ and $D_2$ differ by at most one
element. Namely, $D_2$ is obtained from $D_1$ by adding or removing
one element.

\begin{definition}[$\epsilon$-differential privacy]\label{def:dp}
A randomized mechanism $\mathcal{A} : \D \to \M$ is $\epsilon$-differentially private 
if for all databases $D_1$ and $D_2$ in $\D$ with $D_1 \sim D_2$, and all $S \in \sa$ (the Lebegue $\sigma$-algebra), we have :
\[ \Pb[\mathcal{A}(D_{1})\in S]\leq e^\epsilon \Pb[\mathcal{A}(D_{2})\in S]\]
\end{definition}

\begin{definition}[sensitivity]
The sensitivity $\Delta_f$ of a function $f : \D \to \M$ is
\[\Delta_f=\sup_{D_1,D_2 \in \D, D_1 \sim D_2} d(f(D_1),f(D_2)).\] 
\end{definition}

\subsection{Standard technique to implement differential privacy}
\label{ssect:standard}


The standard technique to grant differential privacy is to add  random noise to the true answer to the query.
In the following, we denote the query by  $f : \D \to \M$. This is usually a deterministic function.
We represent the noise as a random variable $X : \Omega \to \M$.
The standard mechanism, which we will denote by $ \mathcal{A}_0$, returns a probabilistic value which is the sum of the
true result and of a random variable $X$, namely:
\begin{mechanism}\label{mec:additif}
\[ \mathcal{A}_0(D)=f(D)+X \]
\end{mechanism}

\subsection{Error due to the implementation of the query}

The correctness of a mechanism $\mathcal{A}$, if we do not take the implementation
error into account, consists in $\mathcal{A}$ being
$\epsilon$-differentially private.  However, we are  interested
in  analyzing the correctness of the implemented mechanism. 
We start here by discussing the case in which, in mechanism~\ref{mec:additif}, the noise $X$ is 
exact but we take into account the approximation error in the implementation of $f$.

\paragraph{Notation} Given a function $g$, we will indicate by $g'$ 
its implementation, i.e. the function that, for any $x$, gives as result the  value 
actually computed for $g(x)$, with all the approximation and representation errors. \\

The first thing we observe is that 
the implementation of $f$ can give a sensitivity $\Delta_{f'}$ greater than  $\Delta_{f}$ and
we need to take that into account. 
In fact, in the exact semantics the correctness of the mechanism relies on the fact that
$d(f(D_1),f(D_2))\leq \Delta_f$. However, with rounding errors, we may have $d(f'(D_1),f'(D_2)) > \Delta_f$.
Hence we need to require the following  property, usually stronger than differential privacy.

\begin{hypo}\label{correct} Given a mechanism $\mathcal{A}(D) = f'(D)+X$, 
we say that $\mathcal{A}$ satisfies Condition~\ref{correct} with degree $\epsilon$ 
(the desired degree of differential privacy) if
the random variable $X$ has a probability distribution which is absolutely continuous according to the Lebesgue measure,  and  
\[\forall S \in \sa, r_1, r_2 \in \M, 
\Pb[r_1+X \in S]\leq e^{\epsilon \frac{d(r_1,r_2)}{\Delta_{f'}}} \Pb[r_2+X \in S]\]
\end{hypo}

\begin{remark}
In general we expect that  an analysis of the implementation of $f$ will provide some bound on the 
difference between $f$ and $f'$, and that will allow us to provide a bound on $\Delta_{f'}$ in terms of $\Delta_{f}$. 
For instance, if
$\|f-f'\|\leq \delta_f$ then we get $\Delta_{f'}\leq \Delta_{f}+ 2 \delta_f$. \\
\end{remark}

\begin{theorem}\label{laplace0}
Condition~\ref{correct} implies that the mechanism $\mathcal{A}(D) = f'(D)+X$ is $\epsilon$-differentially private (w.r.t. $f'$).
\end{theorem}

The following theorem shows that Condition~\ref{correct} is actually equivalent to differential privacy in the 
  case of    Laplacian noise.
\begin{theorem}\label{Laplace}
Let $\mathcal{A}(D) = f'(D)+X$ be a mechanism, and assume that $X$ is Laplacian. 
If $\mathcal{A}$ is  $\epsilon$-differentially private   (w.r.t. $f'$),  then Condition \ref{correct} holds.
\end{theorem}


%
%

\section{Error due to the  implementation of   the noise}\label{sec:errors}


In this section we consider the implementation error in the noise, trying to make as few assumptions as possible 
about the implementation of real numbers and  of the noise function.

We start with example which shows that any finite implementation makes it impossible for a mechanism 
to achieve the degree of privacy predicted by the theory (i.e. the degree of privacy it has in the exact semantics). 
This example is more general than the one in the introduction in the sense that 
it does not rely on any particular implementation of the real numbers, 
just on the (obvious) assumption that in a physical machine the representation of numbers
in memory is necessarily finite. On the other hand it is less ``dramatic'' than the one in the introduction, because it only shows that the 
theoretical degree of privacy degrades in the implementation, while the example in the introduction shows a case in which $\epsilon$-differential privacy does not hold (in the implementation) for any $\epsilon$. 
 
\begin{example}
{\rm
Consider the standard way to produce a random variable with a given probability law, such as the Laplace distribution.
Randomness on most computers is generated with integers. When we call a function that returns a uniform random value on 
the representation of reals, 
the function generates a random integer $z$ (with uniform law) between $0$ and $N$ (in practice $N \geq 2^{32}$) and returns $u=\nicefrac{z}{N}$. 
From this uniform random generator, we compute $n(\nicefrac{z}{N})$ where $n$ depends on the probability distribution we want to generate.
For instance, to generate the Laplace distribution 
we have $n(u)=- b\,\text{sgn}(u-1/2)\,\ln(1 - 2|u-\nicefrac{1}{2}|)$ which is the inverse of the cumulative function of the Laplace distribution.
However the computation of $n$ is performed in the finite precision semantics, 
i.e. $n$ is a function $\mathbb{F} \to \mathbb{F}$ where $\mathbb{F}$ is the finite set of the representable numbers.
In this setting, the probability of getting some value $x$ for our noise depends on the number of integers $z$ 
such that $n(\nicefrac{z}{N})=x$ :
if there are $k$ values for $z$ such that $n(\nicefrac{z}{N})=x$ then the probability of getting $x$ is $\nicefrac{k}{N}$.
This means, in particular, that, if the theoretical probability for a value $x$ is $1.5/N$, then the closest probability 
actually associated with the drawing of $x$ is either $\nicefrac{1}{N}$ or $\nicefrac{2}{N}$ and in both cases the error is at least $33\%$.
In figure~\ref{ex2}, we illustrate how the error on the distribution breaks the differential-privacy ratio that holds for
the theoretical distribution.
The ratio between the two theoretical Laplacian distributions is  $\nicefrac{4}{3}$.
However, since the actual distribution is issued from a discretization of the uniform generator, the resulting distribution is a step 
function. So when the theoretical probability is very low like in $x_0$, the discretization creates an artificial ratio of $2$ instead of  
$\nicefrac{4}{3}$.
}
\begin{figure}
\centering
  \newcommand*{\sensi}{1}
 \newcommand*{\scale}{0.4}
 \newcommand*{\step}{0.2}
 \newcommand*{\coef}{3}

\begin{tikzpicture}[domain=-6:5]\label{fig-intro}
    \draw[->] (-6.0,0) -- (5.2,0) node[right] {$x$};
    \draw[->] (0,-0.2) -- (0,3.2) node[above] {$P(x)$};
\node at (0,-0.5){$0$};
\foreach \x in {-30,...,25}{
\draw[thick] (\x *\step,{round(\coef*exp(-\scale*abs((\x +0.5)*\step))/\step)*\step}) -- (\x *\step +\step,{round(\coef*exp(-\scale *abs((\x + 0.5)*\step))/\step)*\step});
\draw[thick] (\x *\step,{round(\coef*exp(-\scale*abs((\x -0.5)*\step))/\step)*\step}) -- (\x * \step,{round(\coef*exp(-\scale *abs((\x + 0.5)*\step))/\step)*\step});

\draw[thick,green] (\x *\step,{round(\coef*exp(-\scale*abs((\x +0.5)*\step-\sensi))/\step)*\step}) -- (\x *\step +\step,{round(\coef*exp(-\scale * abs((\x + 0.5)*\step-\sensi))/\step)*\step});
\draw[thick,green] (\x *\step,{round(\coef*exp(-\scale*abs((\x -0.5)*\step-\sensi))/\step)*\step}) -- (\x * \step,{round(\coef*exp(- \scale * abs((\x + 0.5)*\step-\sensi))/\step)*\step});

}
\draw[-] (-5,-0.1) -- (-5,0.5);
\node at (-5,-0.5) {$x_0$};
\end{tikzpicture}

\caption{The probability distributions of Laplace noises generated from a discretized uniform generator}\label{ex2}
\end{figure}
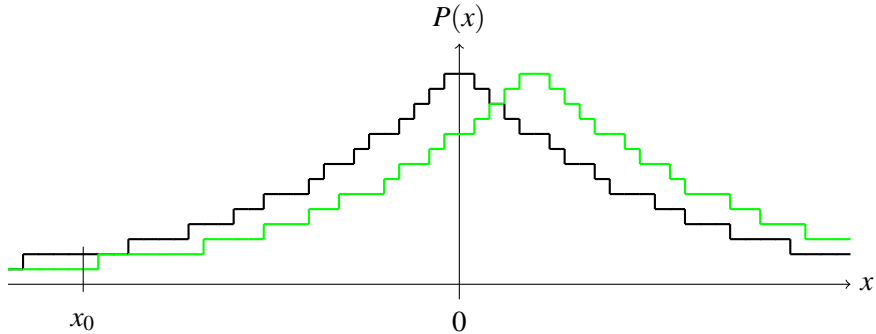
\end{example}


\subsection{The initial uniform random generator}
To generate a random variable, programing languages have only one primitive 
that generates a random value between $0$ and $1$ that aims to
be uniform and independent across several calls.
Hence, to get a random variable with a non uniform distribution, 
we generate it with a function that makes calls to this random generator.
For instance, to draw a value from a random variable $X$ on $\R$ 
distributed according to the cumulative function $C : \R \to ]0,1]$, it is sufficient to 
pick a value $u$ from the uniform generator in $]0,1]$ and then return $C^{-1}(u)$.

We identify three reasons why a uniform random generator may induce
errors.  The first has been explained in the introduction: finite
precision allows generating only $N$ different numbers such that when
we apply a function on the value picked some values are missing and
other are over represented.  The second reason comes from the
generator itself which can returns the $N$ values with different
probabilities even though we might assume that they are returned with
probability $1/N$: furthermore, some values may not even be returned
at all.  A third error is due to the dependence of returned results
when we pick several random values.  Indeed most of the generator
implementations are indeed pseudo generators: when a value is picked
the next one is generated as a hash function of the first one. This
means that if we have $N$ possibilities for one choice then we also have
$N$ possible pairs of successive random values.

To reason about implementation leakage, we have to take into account
all of these sources of errors.  We propose the following model. 
In the exact semantics, the uniform random variable $\U$ 
is generated from a cross product of $\dimU$ uniform independent variables $U$ (with $\dimU \geq \dimension$). 
We denote by $u_1, \dots, u_\dimU \in \Ud$ the values picked by our perfect random generator.
Then we consider the random variable $\U'$ actually provided as
generated from a function $n_0 : \R^\dimU \to \R^\dimU$, $(u'_1, \dots, u'_\dimU) =
n_0(u_1, \dots, u_\dimU)$.
We assume that the bias, i.e., the difference between $n_0$ and the identity, is bounded by 
some $\delta_0 \in \R^+$:
\begin{equation}\label{n0}
\| n_0 - \textrm{Id} \|_\infty \leq \delta_0.
\end{equation}

\subsection{The function $n$ for generating the noise}

From the value $u$ (resp. $u'$) drawn according to the distribution $\U$ in the exact semantics (resp. according to $\U'$ 
in the actual implementation), we generate another value 
by applying the functions $n$  and $n'$ respectively.
Let $X=n(\U)$ be the random variable with the exact distribution and $X'=n'(\U')$ the random variable with the actual one.
\begin{definition}\label{munu}
We denote by $\mu$ and $\nu$ the probability measure of $X$ and $X'$, respectively:
for all $S \in \sa$, $\mu(S)=\Pb[n(\U) \in S]=\lbg(n^{-1}(S))$
and  $\nu(S)=\Pb[n'(\U') \in S]=\lbg(n_0^{-1}(n^{-1}(S)))$. 
\end{definition}

In order to establish a bound on the difference between the probability distribution of $X$ and $X'$  we  
need some condition on the implementation $n'$ of $n$.
For this purpose we use the notion of closeness that we defined in \cite{gazeau12hal}.

\begin{definition}[($k,\delta$)-close,  \cite{gazeau12hal}]\label{closeness}
Let $A$ and $B$ be metric spaces with distance $d_A$ and $d_B$,
respectively. Let $n$ and $n'$ be two functions from $A$ to $B$ and let
$k,\delta \in \R^+$. We say that $n'$ is ($k,\delta$)-close to $n$ if 
\[\forall u,v \in A, d_B(n(u),n'(v)) \leq k\; d_A(u,v) + \delta.\]
\end{definition}

This condition is a combination of 
 the $k$-Lipschitz property, that states a bound between the error on the output and the error on the input, see below, and the implementation errors of $n$:
\begin{definition}[$k$-Lipschitz]
Let $(A,d_A)$ and  $(B,d_B)$ be two metrics spaces and $k \in \R$: 
A function $n : A \to B$ is $k$-Lipschitz if:
\[\forall u, v \in A, d_B(n(u),n(v)) \leq k d_A(u,v) \]
\end{definition}

In \cite{gazeau12hal}, we have proven the following relation between the properties of being $k$-Lipschitz
and of being ($k,\delta$)-close.

\begin{theorem}[\cite{gazeau12hal}]\label{kdelta}
If $n$ is $k$-Lipschitz and $\|n-n'\|_\infty \leq \delta$ then $n$ and $n'$ are ($k,\delta$)-close.
\end{theorem}

We strengthen now the relation by proving that (a sort of) the converse is also true.

\begin{theorem}\label{non-kdelta}
If there exist $u,v$, $d(n(u),n(v)) > k d(u,v) + 2 \delta$ 
then there exist no function $n'$ such that $n$ and $n'$ are ($k,\delta$)-close.
\end{theorem}


Now, we would like $n$ and $n'$ to be ($k,\delta$)-close on $\M$.
However, this implies that $\mathcal{A}_0$ (mechanism~\ref{mec:additif}) cannot be $\epsilon$-differentially-private.
In fact, the latter would imply $\|n(\Ud)\|_\infty=\infty$ otherwise certain answers could be reported (with non-null probability) only 
in correspondence with certain true answers and not with others.
However, $\|n(\Ud)\|_\infty=\infty$ and $\Ud$ bounded implies 
there exist $u,v$, such that $d(n(u),n(v)) > k d(u,v) + 2 \delta$: 
we derive from theorem~\ref{non-kdelta} that $n'$ cannot exist.


In order to keep computed results in a range where we are able to bound the computational errors, one possible solution
consists of a truncation of the result.  
The traditional truncation works as follows: choose a subset $\Mr \subset \M$ 
and, whenever the reported answer $x$ is outside $\Mr$  return  the closest point to $x$ in $\Mr$.
However, while such a procedure is safe in the exact semantics because remapping does not alter  differential privacy, problems might appears when $n$ and $n'$ are not close. 
Furthermore, while in the uni-dimensional case there are two disjoints 
set that are mapped one on the minimal value and the other on the maximal value, 
in higher dimensions we  have a connected set that is mapped on several points, 
and on which the error is not bounded.

Therefore, to remain in a general framework where we do not have any additional knowledge about computational errors for large numbers,
we decide here to return an exception value when the result is outside of some 
compact subset $\Mr$ of $\M$.  
We denote by $\extr$ the value returned by the mechanism when $f'(D)+X' \notin \Mr$. 
Hence, the truncated mechanism $\mathcal{A}$ returns the randomized value or $\extr$:
\begin{mechanism}\label{mec:trunc}
\[\mathcal{A}(D)=\begin{cases}f'(D)+X' & \text{ if }f'(D)+X' \in \Mr \\ \extr & \text{otherwise}\end{cases} \]
\end{mechanism}

We truncate the result because we want to exclude non-robust computations from our mechanism.
However, such a procedure is effective only if unsafe computations remain outside the safe domain.
To grant this property we need two more conditions. 
 One requires the implementation to respect the monotonicity of the computed functions: 
\begin{hypo}\label{monotonicity} We say that a function $g : \M\rightarrow \R^k$ satisfies Condition~\ref{monotonicity} if,
for all $x,y \in \M$, $\|g(x)\| \leq \|g(y)\|$ implies $\|g'(x)\| \leq \|g'(y)\|$.
\end{hypo}
With this property, even if the implementation is not robust for large values, 
if we know some result is not in $\Mr$ then the result for any greater value is not in $\Mr$ either.

The other condition is about the closeness of the implementation of the noise and its exact semantics in a safe area.
For any $\delta_r \in \R^+$, we consider the set $\Ur  \subset \U$ 
defined as $\forall u\in \Ur \; \|n(u)\| \leq \diam(\Mr)+ \delta_r $ i.e. :
$\Ur = n^{-1} \left(\left\lbrace y \left| \|y\| \leq \diam(\Mr) + \delta_r \right. \right\rbrace \right)$.
\begin{hypo}\label{close}
We say that a  noise $n$ satisfies Condition~\ref{close} if $n$ and $n'$ are  ($k,\delta_n$)-close on a set $\Ur$ such that 
\[\forall u \in \Ur^c  \; f'(D) + n(u) \notin \Mr^{+k \delta_0 + \delta_n} \]
\end{hypo}

To find such a set $\Ur$, one possible way is by a   fix point construction.
We begin by finding the smallest $k_0$ and ${\delta_n}_0$ such that $n$ 
and $n'$ are ($k_0,{\delta_n}_0$)-close on $\Mr$.
Then for the generic step $m>0$, we compute the smallest $k_{m+1}$ and ${\delta_n}_{m+1}$ 
such that $n$ and $n'$ are ($k_{m+1},{\delta_n}_{m+1}$)-close on $\Mr^{k_m \delta_0 + {\delta_n}_m}$.


If Conditions~\ref{monotonicity} and \ref{close} hold, then from \eqref{n0} we  derive 
\begin{equation}\label{eq:extr}
\forall u \in \Ur^c  \; f'(D) + n'(n_0(u)) \notin \Mr
\end{equation}
So whatever happens outside of $\Ur$, the result will be truncated.
We can then consider that there is no implementation error outside $\Ur$.
Finally, we have a bound $\delta_t$ for the maximal shift between the exact and the actual semantics:
\begin{equation}\label{epst}
\delta_t =  k\delta_0+\delta_n
\end{equation}

\subsection{A distance between distributions}
Given that we are in a probabilistic setting,  the round-off errors  cannot be measured in terms of numerical difference as they can be in the 
deterministic case,  they should rather be measured in terms of distance between the theoretical distribution and the actual distribution. 
Hence, we need a notion of distance between distributions.
We choose to use the $\infty$-Wassertein distance \cite{champion08} which, as we will show, is the natural metric to measure our deviation.  

\begin{definition}[$\infty$-Wassertein distance]
Let $\mu$, $\nu$ two probability measures on $(\M,\sa)$ such that there exist a compact $\Omega$, $\mu(\Omega)=\nu(\Omega)=1$,
the $\infty$-Wassertein distance between $\mu$ and $\nu$ is defined as follows:
\[W_\infty(\mu,\nu)= \inf_{\gamma \in \Gamma(\mu,\nu)} 
\left(\inf_{ t \geq 0} \left(\gamma\left(\left\lbrace \left.(x, y) \in (\M)^2  \right|  d(x,y)> t   \right\rbrace\right)\right)=0\right)\]
Where $\Gamma(\mu,\nu)$ denotes the collection of all measure on $M\times M$ with marginals $\mu$ and $\nu$ respectively.
\end{definition}
If we denote by $\textrm{Supp}(x,y)$, the support where $\gamma(x,y)$ is non zero, we have
an equivalent definition \cite{champion08} for the $\infty$-Wassertein distance:
\[W_\infty(\mu,\nu)= \inf_{\gamma \in \Gamma(\mu,\nu)} 
\left( \sup_{\textrm{Supp}(x,y)} d(x,y) \right)\]

We extend this definition to any pair of measures that differ only on a compact ($\Mr$ in our case) 
by considering the subset of $\Gamma(\mu,\nu)$ containing only measure $\gamma(x,y)$ 
with $\gamma(x,y) = 0$ if $x\neq y$ and either $x\in \Mr^c$ or $y\in \Mr^c$.

We have introduced this measure because it has a direct link with the computational error as expressed by the following theorem.

\begin{theorem}\label{conv}
Let $X$ and $X'$ be two random variables with distribution $\mu$ and $\nu$ respectively.
We have that $\|X - X' \|_\infty \leq \delta$ implies $d(\mu,\nu)\leq \delta$. 
\end{theorem}

\begin{proof}
We consider the measure $\gamma$ on $M \times M$, $\forall A,B \in \sa, \gamma(A,B)=\Pb(X \in A \AND X' \in B) $.
The marginals of $\gamma$ are $\mu$ and $\nu$. 
Moreover, the support of $\gamma$ is $\delta$ since $\Pb(X \in A \AND X' \in B)=0$ when $A$ and $B$ are distant by more than $\delta$.
Since we have such a $\gamma$ the minimum on all the  $\gamma \in \Gamma(\mu,\nu)$ is less than $\delta$.
\end{proof}

In our case, according to \eqref{epst}, we have $d(\mu,\nu)\leq \delta_t$.
The following theorem allow us to bound the $\mu$ measure of some set with the measure $\nu$. 

\begin{theorem}\label{enc}
\[d(\mu,\nu)\leq \epsilon \implies \forall S \in \M, \nu(S^{-\epsilon}) \leq \mu(S) \leq \nu(S^\epsilon) \] 
\end{theorem}

\begin{proof}
The property of marginals is $\nu(S) = \int_{\M \times S} d\gamma(x,y)$.
Since $\gamma(x,y)=0$ if $d(x,y)> \epsilon$, we derive $\nu(S) = \int_{S^\epsilon \times S} d\gamma(x,y)$.
Then we get $\nu(S) \leq \int_{S^\epsilon \times \M} d\gamma(x,y)$.
The last expression is the marginal of $\gamma$ in $S^\epsilon$, hence by definition of marginal:
$\nu(S) \leq \mu(S^\epsilon)$.
The other inequality is obtain by considering the complement set of $S$ ($\M \setminus S$). 
\end{proof}

\subsection{Rounding the answer}
Once the computation of $\mathcal{A}(D)$ is achieved, we cannot yet return  the answer, because
it could still leak some information.
Indeed, the distribution of $X$ and $X'$ are globally the same, but, on a very small scale, the
 distributions could differ a lot. We prevent this problem by  rounding the result:

\begin{mechanism}\label{mec:round}
The mechanism rounds the result by returning the value closest to $f(D)+n'$ in some discrete subset $S'$.
So $\mathcal{K}(D)=r(\mathcal{A}(D))$ where $r$ is the rounding function. 
\end{mechanism}

From the above rounding function we define the set $\sa'_0$ of all sets that have the same image under $r$. 
Then we define the $\sigma$-algebra $\sa'$ generated by $\sa'_0$: it is the closure under union of all these sets.
Observe now 
that it is not possible for the user to measure the probability that the answer belongs to a set which is not in $\sa'$.
Hence our differential privacy property becomes:  
\begin{equation}\label{s'}
\forall S  \in \sa', \Pb[\mathcal{A}(D_{1})\in S]\leq e^\epsilon \Pb[\mathcal{A}(D_{2})\in S]
\end{equation}

In this way we grant that any measurable set has a minimal measure 
and we prevent the inequality from being violated when probabilities are small.
The following value $R$  represents  the   robustness of the rounding.

\begin{equation}\label{ratio}
\ratio= \max_{S \in \sa'_0, S \neq \emptyset}\frac{\lbg(S^{+\delta_t} \setminus S^{-\delta_t})}{ \lbg(S^{-\delta_t})}
\end{equation}



Now, we are able to prove that if all conditions are met, then the implementation of the mechanism satisfies differential privacy.

\begin{theorem}\label{main}
A $\mathcal{K}$ mechanism that respects Conditions \ref{correct}--\ref{close} is $\epsilon'$-differentially private, with:

\[\forall S \in \sa, \Pb[\mathcal{A}'(D_1)\in S]\leq e^{\epsilon'} \Pb[\mathcal{A}'(D_2)\in S]\]
where $\epsilon' = \epsilon + \ln(1+\ratio e^{\epsilon \frac{\arrondi+\delta_t}{\Delta_{f'}}}) $,
 $\delta_t=k \delta_0 + \delta_n$ and $\arrondi=\max_{S \in \sa'_0} \diam S$.
\end{theorem}

\section{Application to the Laplacian noise in one dimension}\label{sec:one-dim}
In this section we illustrate how to use our result in the  case in which the domain of the answers is $\R$.
The noise added for the protocol, stated in the mechanism~\ref{mec:additif}, is the Laplacian centered in $0$ with scale parameter 
$\Delta_{f'} / \epsilon$.
Theorem \ref{Laplace} implies that Condition \ref{correct} holds for $\epsilon$.
We truncate the result outside of some interval $\Mr=[m,M]$.


\paragraph{Implementation of the $n$ function}
To generate a centered Laplacian distribution from a uniform random variable $U$ in $]0, 1]$, 
a standard method consists in using the inverse of the cumulative function, i.e. 
$X = n(U) = - b\,\text{sgn}(U-1/2)\,\ln(1 - 2|U-1/2|)$,
where $b$ is the intended scale parameter ($\frac{\Delta_{f'}}{\epsilon}$ in our case).
Hence our exact function $n$ is 
\begin{equation}\label{lap1}
 n(u)= \frac{\Delta_{f'}}{\epsilon} \text{sgn}(u-1/2) \ln(1 - 2|u-1/2|).
\end{equation}

\paragraph{Closeness of $n$ and $n'$}
In order to apply our theorem, we need to prove that Condition \ref{close} is satisfied.
By theorem \ref{kdelta}, it is sufficient to prove that, in the 
interval of interest, $n(u)$ is $k$-Lipschitz and that $|n(u)-n'(u)|\leq \delta_n$.
Note that the values of $\delta_n$ and $k$ in general depend on $n$ and on its implementation 
(often the logarithm is  implemented by the CORDIC algorithm).  


The logarithm function used by $n$ is not $k$-Lipschitz for any $k$. 
However, we are interested in the behavior of $n$ when $|n(u)| \leq M-m$.
From the definition of $n$ in \eqref{lap1}, we have:
\[{\frac{dn}{du}}(u) \leq \nmax = \frac{2 \Delta_{f'}}{\epsilon} e^{\frac{\epsilon \diam(\Mr)}{\Delta_{f'}}} \] 
in $\Ur=\{ u | n(u) \leq M-m\}$.
So our function $n$ is $\nmax$-Lipschitz.
Finally, our global error is 
\[\delta_t=\frac{2 \Delta_{f'}}{\epsilon} e^{\frac{\epsilon \diam(\Mr)}{\Delta_{f'}}}\delta_0 + \delta_n \]

\paragraph{Rounding the result}
The rounding process generates a $\sigma$-algebra $\sa'$ composed by small intervals of 
length $\arrondi$ where $\arrondi$ is the accuracy step of the rounding.
In that case, the value defined in (\ref{ratio}) is $\ratio=\frac{\arrondi+2\delta_t}{\arrondi-2\delta_t}$.

%
\paragraph{Differential privacy}
By Theorem \ref{main}, the implementation of our mechanism is $\epsilon'$-differentially private with
\[\epsilon' =  \epsilon + \ln(1+ \frac{\arrondi+2\delta_t}{\arrondi-2\delta_t} e^{\epsilon \frac{\arrondi+\delta_t}{\Delta_{f'}}})\]

\begin{remark}
In case our answer is not in $[m,M]$, we can return $-\infty$ or $+\infty$ instead of $\extr$. 
The reason is that even if the algorithm is not robust when $|u-0.5|$ is small the sign is still correct.
Then we can remap  $-\infty$ to $m$ and $+\infty$ to $M$ to get the usual truncation procedure.
\end{remark}

\section{Application to the Laplacian noise in $\R^2$}\label{sec:two-dim}
When the domain of the answers are the points of a map, like in the 
case of location-based applications,
it is natural to formalize it as the space $\R^2$ equipped with the Euclidean distance.

According to the protocol, we sanitize the results by adding a random variable $X$.
In this case, we will use for $X$ the bivariate Laplacian defined for the Euclidean metric \cite{Mig12} whose 
 density function is:
\[p(x,y)= K e^{b\sqrt{|x-x_0|^2+|y-y_0|^2}}\]
where $K$ is the normalization constant   and $b$ the scale parameter.
Since we are using a Laplacian noise, by Theorem \ref{Laplace},  Condition \ref{correct} holds.

\paragraph{Truncation}
Since most of the time the domain studied is bound 
(for instance the public transportation of a city is inside the limit of the city), 
we can do a truncation. 
However, we recall that our truncation is made for robustness purpose and not just for utility reasons.
Hence, if our domain of interest is a circle, we will not choose $\Mr$ to be the same circle because  
 the probability the truncation would return an exception  would be too high 
 (more than one half if the true result is on the circumference).

\paragraph{Implementation of the $n$ function}
Following \cite{Mig12}, we compute the random variable by drawing  an angle and a distance independently.
The angle $\theta$ is uniformly distributed in $[-\pi,\pi[$.
The radius $r$ has a probability density $D_{\epsilon,R}(r)=\epsilon^2re^{-\epsilon r}$ and  cumulative function 
$ C_\epsilon(r)=1-(1+\epsilon r) e^{-\epsilon r}$.
The radius can therefore be drawn by setting $r=C_\epsilon^{-1}(u)$ where $u$ is generated uniformly in $]0,1]$.

\paragraph{Robustness of $n$}
As in the previous section, we do not analyze an actual implementation but we care about the $k$ factor used for Condition \ref{close}.
First, we analyze for which $k_C(\epsilon,\diam(\Mr))$ the function $C_\epsilon^{-1}$ is $k$-Lipschitz in $[0,\diam(\Mr)]$.
Since $C$ is differential, this question is equivalent to find the inverse of the minimal value  taken by its derivative function on the interval
$C_\epsilon^{-1}([0,\diam(\Mr)])$.
By computing this minimum value, we get:
\[K_C(\epsilon,\diam(\Mr))= \frac{e^{\epsilon \diam(\Mr)}}{2\epsilon+ r \epsilon^2}\]
On the other hand, the computation of $\theta$ is just a multiplication by $2\pi$ of the uniform generator hence
$k_\theta=2\pi$.
Then, with the conversion  $(r,\theta) \mapsto (r \cos(\theta),r \sin(\theta))$ from polar coordinates to Cartesian coordinates we obtain
the global $k$ factor: 
\[k=\sqrt{K_C(\epsilon,\diam(\Mr))^2+2\pi\diam(\Mr)}\]
Let $\delta_n$ be the distance between $n$ and $n'$, and $\delta_0$ be 
the error of the uniform generator. From \eqref{epst} we get:
\[\delta_t=\sqrt{K_C(\epsilon,\diam(\Mr))^2+2\pi\diam(\Mr)} \delta_0 + \delta_n.\]

\paragraph{Rounding the answer}

%
%

%

We now compute the parameter $\ratio$ in \eqref{ratio}.
The rounding is made in the Cartesian coordinates, hence the inverse image of any returned value is a square $S$ of length $\arrondi$.
Note that $S^{\delta_t}$ is included in the square of length $\arrondi+2\delta_t$ and $S^{-\delta_t}$ is a square of length $\arrondi-2\delta_t$.
Hence the ratio value is smaller than $\ratio=(\frac{\arrondi+2\delta_t}{\arrondi-2\delta_t})^2$.
\paragraph{Differential privacy}
By Theorem \ref{main} we get that (the implementation of) our mechanism is $\epsilon'$-differentially private with 
\[\epsilon' =  \epsilon + \ln(1+ (\frac{\arrondi+2\delta_t}{\arrondi-2\delta_t})^2 e^{\epsilon \frac{\arrondi+\delta_t}{\Delta_{f'}}})\]

\section{Conclusion and future work}\label{sec:conclusion}
In this paper we have shown that, in any implementation of mechanisms for differential privacy, 
 the finite precision representation of numbers in any machine 
induces approximation errors that cause the loss of the privacy property.
To solve this problem, we have proposed a method based on   rounding the answer and raising
 an exception when the result is outside some values.
The main result of our paper is that the above method is sound in the sense that 
it preserves differential privacy at the price of a degradation of the privacy degree.
To prove this result, we needed to pay  special attention at expressing the problem 
in terms of probability theory and at defining the link
between computational error and distance between probability distributions.
Finally, we have shown how to apply our method to the case of the linear Laplacian and to that of 
bivariate Laplacian. 

As future developments of this work, we envisage two main lines of research: 
\begin{itemize}
\item Deepening the study of the implementation error in differential privacy:
there are several directions that seem interesting to pursue, including:
\begin{itemize}
\item Improving the mechanisms for generating basic random variables. 
For instance, when generating a one-dimensional random variable, it may have some advantage to pick more values from the uniform random generator, instead than just one (we recall that the standard  method is to draw one uniformly distributed value in $]0,1]$ and then  apply the inverse of the cumulative function). 
For instance, $u_1 + u_2$
has a density function with a triangular shape and cost only one addition. The other advantage is due to the finite representation: if the uniform random generator can pick  $N$ different values then two calls of it generate $N^2$ possibilities, which enlarge considerably the number of possibilities, and therefore reduce the ``holes'' in the distribution.
\item Considering more relaxed versions of differential privacy, for instance the $(\epsilon,\delta)$-differential privacy allows for a (small) additive shift $\delta$ between the two likelihoods in Definition~\ref{def:dp} and it is therefore more tolerant to the implementation error. It would be worth investigating for what values of $\delta$ 
(if any) the standard implementation of differential privacy is safe. 
\end{itemize}
\item Enlarging the scope of this study to the more general area of quantitative information flow. 
There are various notions of information leakage that have been considered in the computer security literature; 
the one considered in differential privacy is just one particular case. Without the pretense of being exhaustive, 
we mention the information-theoretic approaches based on Shannon entropy \cite{Clark:05:QAPL,Malacaria:07:POPL,Chatzikokolakis:08:IC} and those based on R\'enyi min-entropy \cite{Smith:09:FOSSACS,Braun:09:MFPS} and the more recent approach based on decision theory \cite{Alvim:12:CSF}. 
The main difference between differential privacy and these other notions of leakage is that in the former any 
violation of the bound in the likelihood ratio is considered catastrophic, while  the latter focuses on the 
average  amount of leakage, and it is therefore less sensitive to the individual violations. However, even though 
the problem of the implementation error may be attenuated in general by the averaging, we expect that there are cases in which it may still represent a serious problem. 
\end{itemize}

\bibliographystyle{eptcs}
\bibliography{bib,short}
\newpage
\appendix

\section{Proofs of  theorems \ref{laplace0} and \ref{Laplace}}
Theorem~\ref{laplace0}

\begin{proof}
Let $D_1$ and $D_2$ be two databases such that $D_1 \sim D_2$.  Let
$r_1=f'(D_1)$ and $r_2=f'(D_2)$ be two answers.
By definition of sensitivity, $d(r_1,r_2)\leq \Delta_{f'}$ so $e^{\epsilon \frac{d(r_1,r_2)}{\Delta_{f'}}} \leq e^\epsilon$.
Hence, \[ \Pb[\mathcal{A}(D_{1})\in S]\leq e^\epsilon \Pb[\mathcal{A}(D_{2})\in S]\]
\end{proof}

Theorem~\ref{Laplace}

\begin{proof}
First, we show that if $\mathcal{A}$ is $\epsilon$-differentially private 
then $b\leq \frac{\epsilon}{\Delta_{f'}}$ holds for the scale parameter $b$ of $X$. 
Let $D_1\sim D_2$ with $d(f'(D_1),f'(D_2))=\Delta_{f'}$. By $\epsilon$-differential privacy we have, for any 
$S \in \sa$:
\[ \Pb[f'(D_1)+X \in S]\leq e^\epsilon \Pb[f'(D_2)+X \in S]\]
From the density function of the Laplace noise (Definition~\ref{def:Laplace}), we derive:
\[ K(n,d)d\lbg \leq e^\epsilon K(n,d)e^{-b \Delta_{f'}}d\lbg\]
Hence,
\begin{equation}\label{blop}
b\leq \frac{\epsilon}{\Delta_{f'}}.
\end{equation}
Now, by definition of the density function, we have
\[\Pb[r_2+X \in S]= \int_{x \in  S} K(n,d) e^{-b d(x,r_2)} d\lbg\]
From the triangular inequality, we derive:
\[\Pb[r_2+X \in S]\geq \int_{x \in  S} K(n,d) e^{-b (d(x,r_1)+d(r_1,r_2))} d\lbg\]
Hence,
\[\Pb[r_2+X \in S] \geq e^{-b d(r_2,r_1)}\int_{x \in  S} e^{-b d(r_1,x)} d\lbg\]
From inequality \eqref{blop}, we derive:
\[\Pb[r_2+X \in S] \geq e^{-\frac{\epsilon d(r_2,r_1)}{ \Delta_{f'}}} \int_{x \in  S} e^{-b d(r_1,x)} d\lbg\]
Finally,
\[\Pb[r_2+X \in S] \geq  e^{-\frac{\epsilon d(r_2,r_1)}{ \Delta_{f'}}}\Pb[r_1+X \in S] \]
\end{proof}

\section{Proof of the main theorem \ref{main}}
\begin{proof}
Let $S$ in $\sa$. 
%
We first consider the case $S  \neq \extr$.

Define $P_1= \Pb[\mathcal{A}'(D_{1})\in S_a]$ and $P_2=\Pb[\mathcal{A}'(D_{2})\in S_a]$.
Since the result has been rounded (Definition \ref{mec:round}), 
it is equivalent to consider the set $S' \in \sa'$ with $S' = r^{-1}(S)$ instead of $S_a$.

Now we have $P_i=\Pb[f'(D_i) + n'(X) \in S']= \Pb[ n'(X) \in S'-f'(D_i)]$ where $i$ is $1$ or $2$.
Since $\nu$  is the measure associated to $n'$, we have
\[P_i=\nu( S'-f'(D_i))\]
%
%
From \eqref{epst} and Theorem \ref{conv}, $d(\nu,\mu)\leq \delta_t$.
From Theorem \ref{enc} we derive
\[P_1 \leq \mu(S^{\delta_t}-f'(D_1)) \qquad \mbox{and}\qquad P_2 \geq \mu(S^{-\delta_t}-f'(D_2)).\]
The additivity property of measures grants us $\mu(S^{\delta_t})=\mu(S^{-\delta_t})+\mu(S^{\delta_t}-S^{-\delta_t})$.
Condition \ref{correct} can be expressed in term of the measure as:
\[\forall S \in \sa, r \in \M \|r\|,
 \mu(S)\leq e^{\epsilon \frac{\|r\|}{\Delta_{f'}}} \mu(S-r)\]

From this inequality, we can derive, since $\|r\|=\Delta_{f'}$:
\[\mu(S^\epsilon)\leq e^{\epsilon} P_2 + \mu(S^{\delta_t} \setminus S^{-\delta_t})\]

Since  the probability  is absolutely continuous according to the Lebesgue measure (Condition \ref{correct}), 
we can express the probability with a density function $p$:
\[\forall S \in \sa, \mu(S) = \int_S p(x) d\lbg \]
We  derive:
\[ \forall S \in \sa, \min_{x \in S} p(x) \leq \frac{\mu(S)}{ \lbg(S)} \]
By applying this property on $S^{-\delta_t}-f'(D_2)$, we get:
\[\min_{x \in S^{-\delta_t}-f'(D_2)} p(x) \leq \frac{\mu(S^{-\delta_t}-f'(D_2))}{ \lbg(S^{-\delta_t}-f'(D_2))} \]
We   derive: 
\[\exists x_0 \in S-f'(D_2), p(x_0) \leq \frac{P_2}{\lbg(S)} \]

By the triangular inequality, we can bound the distance between $x_0$ and any point of $S^{\delta_t}$ by $\Delta_{f'} + \arrondi + \delta_t$
Hence,  from Condition \ref{correct} we derive:
\[\forall x \in S^{\delta_t}-f'(D_1), p(x) \leq e^{\epsilon\frac{\Delta_{f'}+\arrondi+\delta_t}{\Delta_{f'}}} p(x_0)\]
Then by integration:
\[\mu(S^{\delta_t}-f'(D_1)\setminus S^{-\delta_t}) 
\leq e^{\epsilon\frac{\Delta_{f'}+\arrondi+\delta_t}{\Delta_{f'}}} \frac{\lbg(S^{\delta_t} \setminus S^{-\delta_t})}{\lbg(S^{-\delta_t})}  P_2  \]
We apply the condition \ref{ratio}:
\[\mu(S^{\delta_t}-f'(D_1)\setminus S^{-\delta_t}) \leq e^{\epsilon\frac{\Delta_{f'}+\arrondi+\delta_t}{\Delta_{f'}}} \ratio  P_2 \] 
Finally we obtain :
\[P_1 \leq (1+\ratio e^{\epsilon \frac{\arrondi+\delta_t}{\Delta_{f'}}}) e^{\epsilon}P_2 \]

In case $S$ is $\extr$, due to \eqref{eq:extr},   $\Pb[\mathcal{A}'(D) = \extr]$ is 
the same as
 $\Pb[f'(D)+X' \in \Mr^c]$ where $d(\mu,\nu) \leq \delta_t$. 
Moreover, $\Mr^c$ can be decomposed in a enumerable disjoint union of element of $\sa_0$.
Therefore, the first part of the proof applies:   $\epsilon'$-differential privacy holds for all these elements. 
By the additivity of the measure of disjoint union we conclude.

\end{proof}
\end{document}